\documentclass[letterpaper,journal]{IEEEtran}

\usepackage{amsmath,amsfonts,amssymb,bm}
\usepackage{amsthm}
\usepackage{algorithm}
\usepackage{algpseudocode}
\usepackage{booktabs}
\usepackage{threeparttable}
\usepackage{multirow}
\usepackage{makecell}
\usepackage{array}
\usepackage{graphicx}
\usepackage{cite}
\usepackage{url}
\usepackage{microtype}
\usepackage[hidelinks]{hyperref}

\newtheorem{definition}{Definition}
\newtheorem{lemma}{Lemma}
\newtheorem{theorem}{Theorem}

\newcommand{\SamplerZ}{\textsc{SamplerZ}}
\newcommand{\BaseSampler}{\textsc{BaseSampler}}
\newcommand{\BerExp}{\textsc{BerExp}}

\newcommand{\Dimpl}{\widehat{D}}

\begin{document}

\title{Bi-SamplerZ: A Rejection-Aware Cooperative Gaussian Sampling Framework for Falcon Signature Hardware}

\author{Binke~Zhao,
        Ghada~Alsuhli,
        Hani~Saleh,~\IEEEmembership{Senior Member,~IEEE},
        and~Baker~Mohammad,~\IEEEmembership{Senior Member,~IEEE}%
\thanks{B. Zhao is with The Hong Kong University of Science and Technology (Guangzhou), Guangzhou, China (e-mail: bzhao390@connect.hkust-gz.edu.cn). This work was carried out while he was a research assistant at the System on Chip Center, Khalifa University, Abu Dhabi, UAE.}%
\thanks{G. Alsuhli is with the Faculty of Engineering \& IT, The British University in Dubai, Dubai, UAE (e-mail: ghada.alsuhli@buid.ac.ae). H. Saleh and B. Mohammad are with the Department of Electrical and Computer Engineering and the System on Chip Center, Khalifa University, Abu Dhabi, UAE. G. Alsuhli is the corresponding author.}%
}

\maketitle

\begin{abstract}
Falcon, selected by NIST as the basis of FN-DSA, provides compact lattice-based signatures and efficient verification, but its signing procedure relies on a run-time-parameterized discrete Gaussian sampler whose rejection loop produces irregular latency. This paper identifies an architectural property of Falcon's fast Fourier sampling procedure that is not exploited by a conventional sampler organization: every scalar recursion leaf naturally issues two independent \SamplerZ\ requests with different centers and a shared standard deviation. Static dual-path execution exposes this parallelism, yet it also creates a recurrent imbalance---when one request is accepted before the other, the completed datapath becomes idle while the rejected request continues.
We present \emph{Bi-SamplerZ}, a rejection-aware cooperative sampling framework that converts this idle capacity into useful computation. After an asymmetric accept/reject outcome, Bi-SamplerZ latches the completed logical result and dynamically reassigns the released physical datapath to the unfinished target. The two paths then evaluate fresh independent candidates for the same remaining distribution. We show that this post-rejection cooperation increases the assisted-round completion probability without modifying the underlying candidate distribution or Bernoulli acceptance rule, and we state the randomness-allocation conditions required to preserve the joint output distribution of the original pair of logical sampler calls. The framework is realized with shared pair-level preprocessing, two rejection datapaths, explicit random-word management, and supporting arithmetic optimizations. Compared with the evaluated FalconSign sampler baseline, Bi-SamplerZ reduces the rejection-free two-sample cycle count from 137 to 59 and the modeled expected pair latency from 230.83 to 106.08 cycles, a 54.0\% cycle reduction. Within Bi-SamplerZ itself, cooperative assistance further reduces the expected latency from 111.54 to 106.08 cycles despite a nine-cycle retargeting cost. FPGA and ASIC implementations demonstrate the practicality of the proposed framework and its compatibility with high-performance Falcon/FN-DSA signing hardware.
\end{abstract}


\section{Introduction}
\label{sec:introduction}

\IEEEPARstart{P}{ost}-quantum cryptography is moving from algorithm selection toward deployable standards and efficient implementations, creating new challenges for hardware designers~\cite{sakwa2026survey,yan2026lattice}. Lattice-based schemes, which underlie most of NIST's selected standards, replace the modular exponentiation and elliptic-curve arithmetic of classical public-key cryptography with polynomial-ring operations, number-theoretic transforms, and, in several schemes, Gaussian sampling, primitives whose efficient hardware realization has become an active research area \cite{alshehhi2026design, alsuhli2025efficient, smole2026giving, luo2026rethinking, he2021novel}.

Falcon is a representative and particularly demanding example: selected by the National Institute of Standards and Technology (NIST) as the basis of FN-DSA~\cite{nist_selected,nist_fips206}, it is built on NTRU lattices and fast Fourier sampling (FFS) to provide compact public keys and signatures together with efficient verification~\cite{Falcon,gpv2008}. Its signing path, however, is substantially more demanding than its verification path, combining recursive FFT-domain computation, precision-sensitive arithmetic, pseudorandom generation, and discrete Gaussian sampling with run-time-dependent parameters.

Within this signing path, \SamplerZ, the scalar rejection sampler invoked at the leaves of FFS, is the component most exposed to run-time variability \cite{smole2026giving}. Each attempt generates a candidate from a base distribution, assigns its sign, evaluates an exponential term, and performs a Bernoulli acceptance test. Rejection restarts the attempt, so the execution time of one logical sample is inherently variable. Existing Falcon hardware has attacked the cost of this operation through HW/SW partitioning~\cite{Karabulut,Lee}, vector extensions~\cite{Yu}, system-level scheduling and optimized sampler pipelines~\cite{falconsign}, and increasingly aggressive attempt-level arithmetic such as parallel sampling loops, unrolled exponential evaluation, and constant-latency Bernoulli processing~\cite{pendyala2025}. Recent work also includes compact two-sample implementations~\cite{vojacek2025} and a prediction-based dual-channel sampler in FALCON-Wings~\cite{falconwings2026}. These developments substantially move the absolute performance frontier, making global ``fastest'' or ``best'' claims from earlier Falcon hardware rapidly obsolete.

Unlike a generic Gaussian-noise generator, Falcon's leaf sampler is protocol-defined: its run-time parameters, candidate construction, and Bernoulli rejection rule are part of the signing procedure~\cite{Falcon}. Consequently, Bi-SamplerZ does not replace \SamplerZ\ with an alternative Gaussian-generation method; it preserves the implemented sampling rule and optimizes how the two protocol-mandated logical calls are scheduled onto hardware resources.

The architectural observation behind Bi-SamplerZ is orthogonal to that performance race. At every scalar FFS leaf, Falcon issues \emph{two} \SamplerZ\ calls for distinct centers but a common leaf-dependent standard deviation. Once the leaf inputs have been formed, these requests are logically independent and can execute concurrently. A straightforward dual-sampler therefore maps one request to each physical datapath. Rejection sampling, however, makes their completion times stochastic. When exactly one path accepts, the accepted logical sample must be retained, but the hardware that produced it no longer has useful work under static one-to-one assignment. The other path may spend one or more additional rounds completing the rejected request.

Bi-SamplerZ treats this asymmetric completion event as a scheduling opportunity rather than unavoidable idle time. It introduces a \emph{rejection-aware cooperative sampling framework}: after one logical request completes, the result is frozen and the released physical path is retargeted to the unfinished request. Both physical paths then generate independent fresh attempts for the same remaining target until one succeeds. The mechanism is reactive rather than predictive---it is activated by an observed accept/reject imbalance---and it changes only the mapping between physical resources and logical requests. It does not modify the base distribution, the candidate transformation, or the Bernoulli rejection rule.

This perspective changes the main claim of the work. Bi-SamplerZ is not presented as the current fastest Falcon sampler simply because it contains two datapaths. Its contribution is a pair-aware transaction framework that exploits Falcon's leaf-level call structure and the stochastic imbalance induced by rejection sampling. The same scheduling principle can surround different arithmetic cores, including future faster or standards-conforming sampler implementations.

The main contributions are summarized as follows.
\begin{itemize}
    \item \textbf{Paired-call-aware sampling framework:} We expose the two leaf-level \SamplerZ\ calls as one shared-context hardware transaction and identify asymmetric rejection outcomes as a frequent source of otherwise idle compute capacity.

    \item \textbf{Rejection-aware cooperative scheduling:} We introduce a dynamic reassignment mechanism that latches the completed logical sample and redirects the released lane to the unfinished target. With independent attempts, the assisted-round completion probability becomes $1-(1-p)^2$ for a target whose single-attempt success probability is $p$.

    \item \textbf{Distribution-preserving control formulation:} We state the conditions under which cooperation preserves the implemented sampler distribution: helper attempts use the same target parameters, consume fresh non-overlapping randomness, and resolve simultaneous acceptances through a value-independent priority rule.

    \item \textbf{Hardware realization and quantitative validation:} We implement the framework with shared preprocessing, two rejection pipelines, ChaCha20-backed random refill, and hardware-oriented arithmetic optimizations on FPGA and ASIC. Relative to the evaluated FalconSign sampler, the complete design reduces the modeled expected pair cycle count from 230.83 to 106.08 cycles (54.0\%). An internal ablation further shows that the assistance mechanism alone reduces expected pair latency from 111.54 to 106.08 cycles (4.90\%) despite the nine-cycle switching overhead.
\end{itemize}

The remainder of the paper first motivates the pair-level execution model, then presents the cooperative protocol and its correctness conditions, describes the hardware architecture and supporting microarchitectural optimizations, evaluates the resulting implementation, and finally positions the framework relative to recent Falcon hardware. The implementation repository is available at \url{https://github.com/Shaibk/Bi-SamplerZ}.

\section{Background and Motivation}
\label{sec:background}\label{sec:Background}

\subsection{Falcon/FN-DSA Sampling Context}

Falcon follows the GPV hash-and-sign framework over NTRU lattices~\cite{gpv2008,Falcon}. Efficient discrete Gaussian sampling has long been an implementation concern in lattice cryptography, motivating CDT-based, table-driven, and configurable hardware designs across schemes such as BLISS, LP-family constructions, qTESLA, and related lattice systems~\cite{blss,lp2,qtesla,design-time}. Falcon is more demanding because its signing path invokes a sampler with run-time-dependent centers and scale parameters rather than a single fixed distribution. During signing, FFS recursively samples a short lattice vector in the Fourier domain. A simplified view of the recursion leaf is shown in Algorithm~\ref{alg:ffs_leaf}. At $n=1$, two scalar values $t_0$ and $t_1$ are available, and the same leaf stores one standard deviation $\sigma'$. The sampler is called once for each center.

\begin{algorithm}[t]
\caption{Relevant Leaf of Falcon Fast Fourier Sampling~\cite{Falcon}}
\label{alg:ffs_leaf}\label{ffsampling}
\begin{algorithmic}[1]
\Require Scalar leaf centers $t_0,t_1$ and leaf standard deviation $\sigma'$
\Ensure Integer pair $(z_0,z_1)$
\State $z_0 \gets \SamplerZ(t_0,\sigma')$
\State $z_1 \gets \SamplerZ(t_1,\sigma')$
\State \Return $(z_0,z_1)$
\end{algorithmic}
\end{algorithm}

The pair is a valid local parallelization boundary. Deeper recursive branches cannot in general be executed arbitrarily in parallel because the sampled result of one branch participates in the center calculation of another. At a scalar leaf, by contrast, both centers have already been formed. Conditioned on $(t_0,t_1,\sigma')$, the two logical sampler calls use independent randomness and produce the pair required by the parent recursion.

NIST has continued the development of FIPS~206/FN-DSA around Falcon and has highlighted floating-point and implementation issues as part of that standardization effort~\cite{nist_fips206}. Accordingly, this work focuses on the Falcon sampling structure implemented by the evaluated RTL and treats the cooperative scheduler separately from arithmetic-conformance questions. Recent analyses of Falcon's numerical behavior and formal security further reinforce this separation~\cite{lin2025,fouque2026}.

\subsection{SamplerZ Rejection Procedure}

Algorithm~\ref{alg:samplerz} summarizes the Falcon-style scalar sampler. It first decomposes $\mu$ into an integer part and a fractional part $r$. Each loop iteration draws a magnitude $z_0$ from a reverse cumulative distribution table (RCDT), uses a random sign bit $b$ to obtain a signed candidate $z$, evaluates an exponent $x$, and accepts according to \BerExp. The returned value is distributed according to the concrete sampler implemented by the selected arithmetic and approximation path.

\begin{algorithm}[t]
\caption{$\SamplerZ(\mu,\sigma')$~\cite{Falcon}}
\label{alg:samplerz}\label{samplerz}
\begin{algorithmic}[1]
\Require Center $\mu$ and $\sigma'\in[\sigma_{\min},\sigma_{\max}]$
\Ensure Integer output following the implemented sampler distribution
\State $r\gets\mu-\lfloor\mu\rfloor$
\State $ccs\gets\sigma_{\min}/\sigma'$
\While{true}
    \State $z_0\gets\BaseSampler()$
    \State $b\gets\textsc{UniformBit}()$
    \State $z\gets b+(2b-1)z_0$
    \State $x\gets\dfrac{(z-r)^2}{2\sigma'^2}-\dfrac{z_0^2}{2\sigma_{\max}^2}$
    \If{$\BerExp(x,ccs)=1$}
        \State \Return $z+\lfloor\mu\rfloor$
    \EndIf
\EndWhile
\end{algorithmic}
\end{algorithm}

The reference-style \BerExp\ decomposes $x=s\ln 2+r_x$, approximates $ccs\exp(-r_x)$ through a fixed polynomial recurrence, scales by $2^{-s}$, and compares the resulting threshold with uniformly generated bytes~\cite{Falcon}. The arithmetic latency of this path is substantial; more importantly for this work, the number of complete attempts is random.

\subsection{Pair-Level Completion Imbalance}
\label{subsec:imbalance}

Table~\ref{tab:failstats} reproduces the 102,400-call failure histogram reported in the EFX study and used by our latency model~\cite{Lee}. A ``failure'' denotes a rejected candidate before the final accepted candidate. The first-attempt success fraction is $p=0.5758$, whereas $42.42\%$ of calls require at least one retry.

\begin{table}[t]
\caption{SamplerZ Failure Histogram Used for Motivation and Modeling~\cite{Lee}}
\label{tab:failstats}\label{tab:samplerz-failure}
\centering
\setlength{\tabcolsep}{2.6pt}
\renewcommand{\arraystretch}{1.10}
\resizebox{\columnwidth}{!}{%
\begin{tabular}{c|rrrrrrr|r}
\toprule
Failures & 0 & 1 & 2 & 3 & 4 & 5 & $\geq6$ & Total \\
\midrule
Calls & 58,957 & 24,978 & 10,725 & 4,465 & 1,931 & 742 & 602 & 102,400 \\
Ratio (\%) & 57.58 & 24.39 & 10.47 & 4.36 & 1.89 & 0.72 & 0.59 & 100 \\
\bottomrule
\end{tabular}}
\end{table}

Let $p_l$ and $p_r$ denote the success probabilities of one current attempt for the left and right requests. A normal dual-path round produces
\begin{align}
P_{AA} &= p_l p_r,\\
P_{RR} &= (1-p_l)(1-p_r),\\
P_{AR} &= p_l(1-p_r)+(1-p_l)p_r,
\label{eq:asym_general}
\end{align}
where $P_{AR}$ is the probability of an asymmetric accept/reject outcome. Using the empirical average $p_l=p_r=p=0.5758$ only as a first-order illustration gives
\begin{align}
P_{AA}&=0.3315, & P_{RR}&=0.1799, & P_{AR}&=0.4885.
\end{align}
Thus, almost half of pair rounds are asymmetric in this average model. Figure~\ref{fig:motivation} contrasts static and cooperative handling of that event.

\begin{figure}[t]
\centering
\includegraphics[width=\columnwidth]{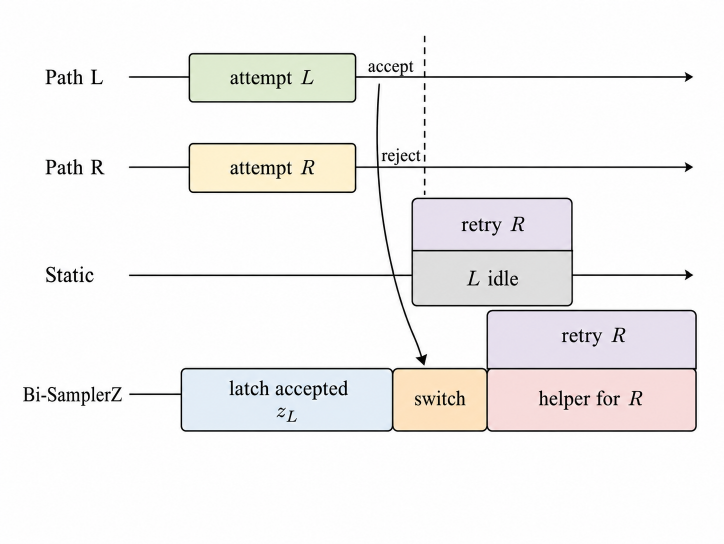}
\caption{After an asymmetric outcome, a static pair engine idles the completed path. Bi-SamplerZ freezes the accepted result and reuses that path as a helper for the unfinished target.}
\label{fig:motivation}
\end{figure}

This observation defines the paper's problem statement: given that a two-path sampler has already paid for parallel hardware, how can it reduce the completion tail of the unfinished request without changing the behavior of two logical \SamplerZ\ calls?

\section{Rejection-Aware Cooperative Sampling}
\label{sec:cooperative}

\subsection{Assistance Protocol}

Bi-SamplerZ begins every transaction in \emph{normal mode}, assigning request $L=(\mu_l,\sigma')$ to the left datapath and request $R=(\mu_r,\sigma')$ to the right datapath. The controller resolves each round according to Algorithm~\ref{alg:assist}. An accepted result is immutable once latched. Assistance is entered only when one valid result and one unfinished target coexist.

\begin{algorithm}[t]
\caption{Pair Controller for Rejection-Aware Assistance}
\label{alg:assist}
\begin{algorithmic}[1]
\Require Pair $(\mu_l,\mu_r,\sigma')$
\State Assign path $l$ to target $L$ and path $r$ to target $R$
\Repeat
    \State Execute one fresh attempt on every active path
    \If{both logical targets are accepted}
        \State \Return latched pair $(z_l,z_r)$
    \ElsIf{neither logical target is accepted}
        \State Keep the normal assignments and retry
    \ElsIf{exactly one target is accepted and not yet in assist mode}
        \State Latch the accepted result
        \State Copy the unfinished target parameters to the released path
        \State Enter assist mode after the switch sequence
    \Else
        \State Both paths remain assigned to the unfinished target
        \State Latch the result from the first accepted lane by fixed priority
    \EndIf
\Until{both logical targets are accepted}
\end{algorithmic}
\end{algorithm}

In assist mode, the two physical paths no longer correspond one-to-one with the two logical outputs. They instead issue two independent attempts for the same unfinished target. A fixed lane-priority rule resolves the rare case in which both helper attempts accept in the same cycle. The rule is independent of candidate values and therefore does not bias selection.

Figure~\ref{fig:assistflow} makes the logical/physical distinction explicit. The stored result keeps its original output tag, while the unfinished output tag is associated with both active lanes. This tag-based view is useful for controller verification because it prevents a helper result from being written into the already completed logical slot.

\begin{figure*}[t]
\centering
\includegraphics[width=0.98\textwidth]{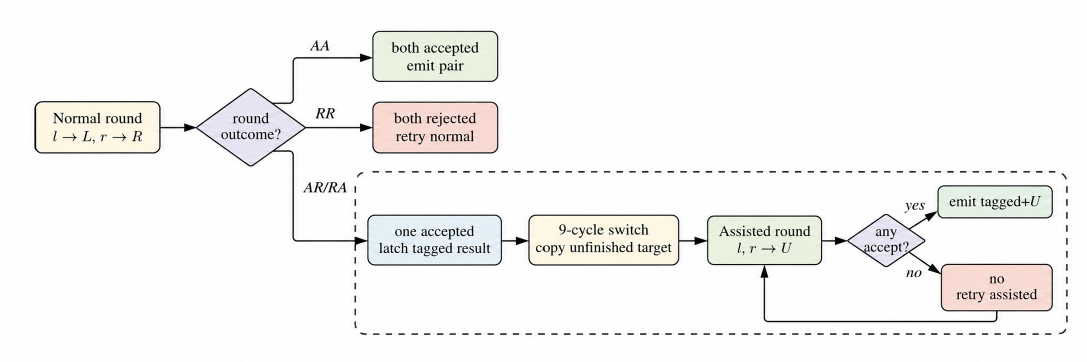}
\caption{Logical flow of Bi-SamplerZ. Assistance is triggered by an observed asymmetric outcome rather than by predicting a rejection. The completed logical output is stored while both physical paths target the unfinished request.}
\label{fig:assistflow}
\end{figure*}

\subsection{Assisted-Round Completion Probability}

For an unfinished target whose current independent attempt succeeds with probability $p$, a conventional single active path completes the next round with probability $p$. Two independent active paths complete the assisted round if at least one accepts:
\begin{equation}
q = 1-(1-p)^2=p(2-p).
\label{eq:q}
\end{equation}
For the empirical average $p=0.5758$, $q=0.8201$. The quantity $q$ is an \emph{assisted-round completion probability}; it is not a modified acceptance probability for an individual Falcon candidate.

Let $L$ denote the average latency of an additional rejection round and $S$ the one-time retargeting cost. In an idealized memoryless model, the expected additional time for one active lane is $L/p$, whereas two assisted lanes require $S+L/q$. Assistance is beneficial when
\begin{equation}
S < L\left(\frac{1}{p}-\frac{1}{q}\right)
  = L\frac{1-p}{p(2-p)}.
\label{eq:threshold}
\end{equation}
At $p=0.5758$, the right-hand side is approximately $0.5173L$. The implemented switch sequence uses $S=9$ cycles, while the rejection-evaluation portion occupies 34--41 cycles. Equation~\eqref{eq:threshold} therefore predicts a favorable regime. This first-order result explains the mechanism; the complete latency calculation in Section~\ref{sec:evaluation} uses the full failure histogram and actual state costs rather than assuming a single geometric parameter.

\subsection{Correctness Relative to the Implemented Sampler}
\label{subsec:correctness}

A scheduling optimization for a cryptographic sampler must not be justified solely by latency. We therefore state the distributional assumptions explicitly. Let $\theta=(\mu,\sigma')$ denote a target parameter pair. Because the concrete arithmetic and polynomial approximation may differ across implementations, let $\Dimpl_{\theta}$ denote the output distribution of one complete baseline \SamplerZ\ call implemented with a fixed arithmetic realization. The claim below is equality to $\Dimpl_{\theta}$, not a new claim that the arithmetic realization exactly equals an ideal real-valued Gaussian.

\begin{definition}[Independent attempt]
An attempt for $\theta$ consumes a fresh random segment $\rho$ and returns an acceptance flag $A\in\{0,1\}$ and a candidate $X$. For a fixed implementation, $\Pr[A=1]=p_{\theta}$ and
\begin{equation}
\Pr[X\in\mathcal{S}\mid A=1]=\Dimpl_{\theta}(\mathcal{S})
\end{equation}
for every set of integer outputs $\mathcal{S}$.
\end{definition}

\begin{lemma}[Fixed-priority parallel attempts]
\label{lem:parallel}
Evaluate $K\geq1$ independent attempts for the same $\theta$, and let event $C$ denote that at least one attempt accepts. If the controller returns the lowest-index accepted attempt, then
\begin{equation}
\Pr[Y\in\mathcal{S}\mid C]=\Dimpl_{\theta}(\mathcal{S}).
\end{equation}
\end{lemma}

\begin{proof}
Index attempts by $1,\ldots,K$ and write $p=p_{\theta}$. The event that lane $i$ is selected requires lanes $1,\ldots,i-1$ to reject and lane $i$ to accept; later outcomes are irrelevant. Independence gives
\begin{align}
\Pr[Y\in\mathcal{S},C]
 &= \sum_{i=1}^{K}(1-p)^{i-1}p\,\Dimpl_{\theta}(\mathcal{S})\\
 &= \left(1-(1-p)^K\right)\Dimpl_{\theta}(\mathcal{S}).
\end{align}
Since $\Pr[C]=1-(1-p)^K$, division by $\Pr[C]$ yields the result.
\end{proof}

\begin{theorem}[Pair distribution preservation]
\label{thm:pair}
Consider logical targets $\theta_l$ and $\theta_r$. Assume: 1) every physical attempt implements the same baseline rules for its assigned target; 2) candidate, sign, and Bernoulli random words are non-overlapping across attempts; 3) helper attempts use only fresh randomness generated after or disjoint from the stored result; and 4) simultaneous helper acceptances are resolved by a value-independent fixed priority. Then Bi-SamplerZ returns the same joint distribution as two independent baseline calls:
\begin{equation}
(Z_l,Z_r)\sim \Dimpl_{\theta_l}\times\Dimpl_{\theta_r}.
\end{equation}
\end{theorem}

\begin{proof}
In normal mode, the two paths use disjoint randomness. If both accept in the same round, their conditional outputs are independent samples from $\Dimpl_{\theta_l}$ and $\Dimpl_{\theta_r}$. If both reject, no output is exposed and the same argument applies inductively to the next fresh round.

Suppose exactly one target, without loss of generality $\theta_l$, accepts. Its accepted value follows $\Dimpl_{\theta_l}$ and is independent of the right path's rejection event because the paths use disjoint randomness. The stored value is never changed. All subsequent attempts target $\theta_r$ and use fresh random segments independent of the stored value. By Lemma~\ref{lem:parallel}, each assisted round that completes returns a value distributed as $\Dimpl_{\theta_r}$; rounds in which both attempts reject expose no value and can be repeated inductively. Therefore the final right output follows $\Dimpl_{\theta_r}$ and is independent of the stored left output. The symmetric asymmetric case is identical. Mixing over the mutually exclusive normal-mode outcomes preserves the product distribution.
\end{proof}

The theorem identifies implementation obligations that are checkable at the RTL level: target tags must be preserved, accepted registers must be write-protected, and the refill controller must never allocate the same PRNG word to two logical attempts. It also separates scheduler correctness from arithmetic-distribution validation. Because cooperation does not alter candidate generation or the Bernoulli rule, any goodness-of-fit, tail-accuracy, or numerical-conformance property of the baseline realization is inherited unchanged by the scheduler; conversely, Theorem~\ref{thm:pair} does not by itself certify that $\Dimpl$ is identical to an ideal real-valued Gaussian or to a future final FN-DSA arithmetic profile.

Bi-SamplerZ remains a variable-latency rejection sampler. The number of rejection rounds and the transition into assistance may therefore be observable through timing or power, so the present work does not claim constant-time or side-channel resistance. Countermeasures for such leakage are orthogonal to the transaction-level scheduling contribution and are left to a separately protected integration.

\section{Bi-SamplerZ Microarchitecture}
\label{sec:architecture}

\subsection{Overall Organization}

Figure~\ref{fig:samplerz_architecture} shows the pair engine. The architecture separates pair-level work, candidate-level work, and control. A shared front end receives $(\mu_l,\mu_r,\sigma')$, extracts the fractional centers, and prepares constants common to both logical requests. A shared RCDT structure and random refill subsystem feed two candidate/rejection datapaths. Each path contains candidate formation (\texttt{Bef\_loop}), the polynomial recurrence (\texttt{For\_loop}), and the final byte-wise comparison (\texttt{CMP}). Accepted results are written into tagged output registers. A cross-assignment network allows either physical path to read the unfinished target state during assistance.

\begin{figure*}[!t]
\centering
\includegraphics[width=1\textwidth]{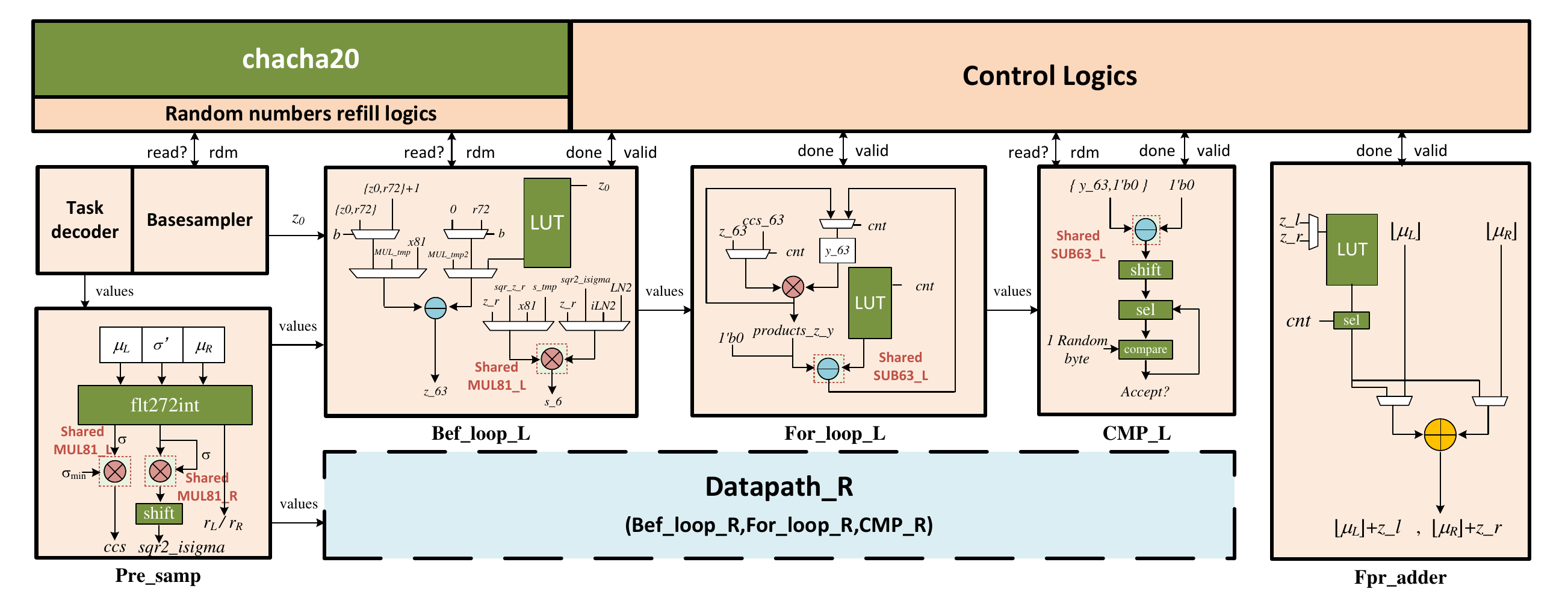}
\caption{Overview of the Bi-SamplerZ functional decomposition and datapath connections.}
\label{fig:samplerz_architecture}
\end{figure*}

Treating the leaf calls as one transaction permits three forms of sharing. First, $\sigma'$-dependent preprocessing is performed once. Second, the two \BaseSampler\ lanes share constant storage and selection organization while consuming distinct random words. Third, a single controller and refill scheduler coordinate both requests and the helper mode. No candidate-dependent \BaseSampler\ output, sign bit, exponent value, or \BerExp\ state is shared between logical attempts. The long candidate-dependent recurrence remains duplicated because simultaneous independent attempts are the source of the latency reduction.

\subsection{Functional Pipeline}

The datapath is partitioned according to data dependence rather than software routine boundaries.

\textbf{Pre\_samp:} For every pair, the front end computes
\begin{align}
r_l&=\mu_l-\lfloor\mu_l\rfloor, &
 r_r&=\mu_r-\lfloor\mu_r\rfloor,\\
ccs&=\sigma_{\min}/\sigma', &
 h_{\sigma}&=1/(2\sigma'^2).
\end{align}
The terms $ccs$ and $h_{\sigma}$ are shared by the two logical requests.

\textbf{Bef\_loop:} Each path consumes a magnitude $z_0$, sign bit $b$, target fractional center $r$, and the shared scale. It forms $z$ and the exponent input
\begin{equation}
x=(z-r)^2h_{\sigma}-T[z_0],
\quad T[z_0]=\frac{z_0^2}{2\sigma_{\max}^2}.
\label{eq:xtable}
\end{equation}
It then decomposes $x=s\ln 2+r_x$ and initializes the recurrence state.

\textbf{For\_loop:} The reference-style polynomial approximation is evaluated through the recurrence
\begin{equation}
y_i=C_i-\left(z_x y_{i-1}\gg 63\right),\qquad i=1,\ldots,12,
\label{eq:recurrence}
\end{equation}
followed by the $ccs$ scaling. Isolating this dependency chain allows candidate preparation for a possible future round to overlap with the current long recurrence where buffer capacity permits.

\textbf{CMP and final addition:} The scaled threshold is compared byte by byte with fresh uniform data. Acceptance sets the tag-associated valid bit; the output stage adds the signed offset to $\lfloor\mu\rfloor$ for the corresponding logical target.

\subsection{Controller and State Transitions}

The finite-state machine in Fig.~\ref{fig:State Machine of Bi-SamplerZ} separates normal and assisted operation. \texttt{PRE} performs pair-level preparation; \texttt{NREG} decouples front-end and loop timing; \texttt{NLOOP} executes one attempt for each normal target; \texttt{SWITCHL}/\texttt{SWITCHR} rebuild candidate-dependent state for the newly assigned helper; \texttt{ALOOP} evaluates two attempts for the unfinished target; and \texttt{F\_ADD} commits the pair.

\begin{figure*}[!t]
\centering
\includegraphics[width=0.95\textwidth]{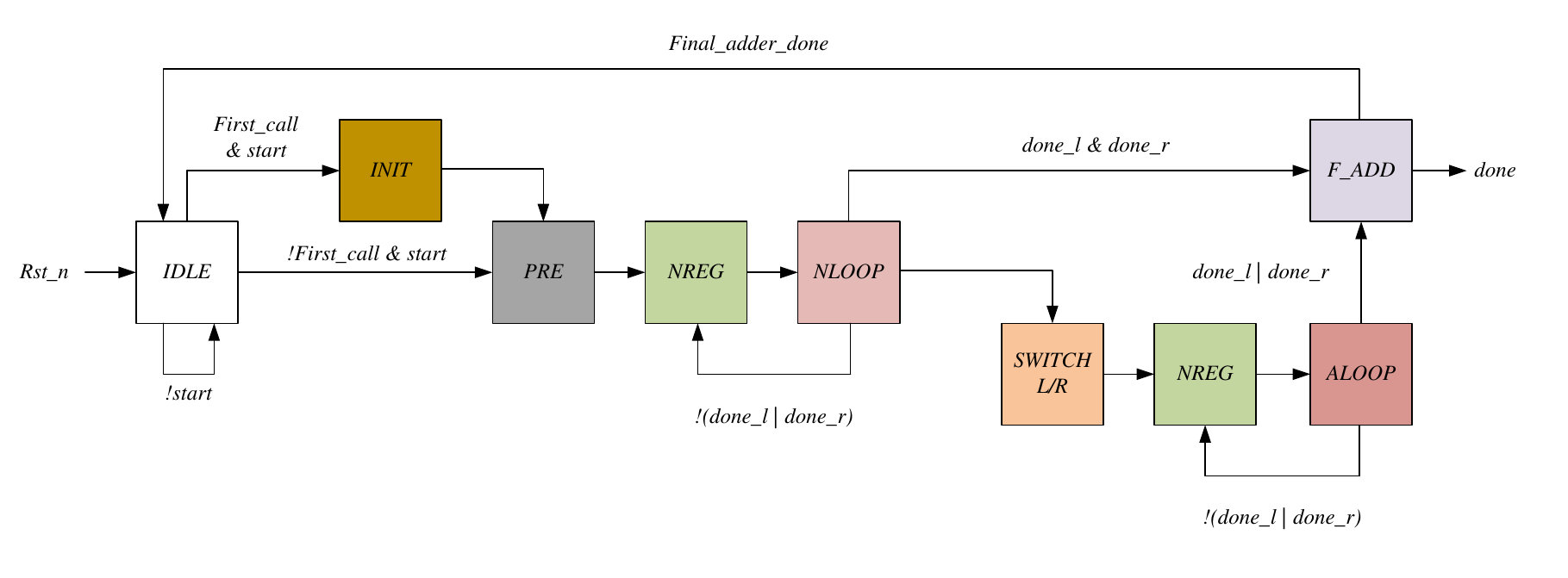}
\caption{State-transition diagram of Bi-SamplerZ.}
\label{fig:State Machine of Bi-SamplerZ}
\end{figure*}

Table~\ref{tab:states} lists the architectural responsibilities of each state. The switch phase costs nine cycles in the evaluated RTL. This latency includes selecting the unfinished center, regenerating the path-local pre-loop values, and registering the helper context. The accepted output register is not on the helper feedback path and remains write-disabled throughout assistance.

\begin{table}[t]
\caption{Controller Responsibilities}
\label{tab:states}\label{tab:state_module_activation}
\centering
\setlength{\tabcolsep}{3.4pt}
\renewcommand{\arraystretch}{1.08}
\begin{tabular}{p{0.19\columnwidth}p{0.72\columnwidth}}
\toprule
State & Primary operation \\
\midrule
IDLE & Wait for a paired request; preserve PRNG/refill state. \\
PRE & Compute pair-invariant and first-candidate quantities. \\
NREG & Register intermediate state and isolate pipeline stages. \\
NLOOP & Execute one attempt for each unfinished logical target. \\
SWITCHL/R & Retarget the completed physical path to the unfinished target. \\
ALOOP & Execute two independent attempts for one unfinished target. \\
F\_ADD & Commit tagged offsets and assert the pair-valid output. \\
\bottomrule
\end{tabular}
\end{table}

\subsection{Random-Word Allocation}
\label{subsec:random}

Correct assistance requires more than two arithmetic datapaths. The random subsystem must preserve attempt independence across \BaseSampler\ inputs, sign bits, and byte-wise \BerExp\ comparisons. The implementation uses ChaCha20-backed refill buffers and advances a word pointer whenever an attempt reserves data. Reserved words are tagged with the destination path and purpose; a rejected attempt never returns its words to the pool. During switching, the helper path obtains a new allocation rather than copying the failed path's randomness.

This discipline serves two goals. First, it satisfies the disjoint-randomness assumptions of Theorem~\ref{thm:pair}. Second, it prevents speculative preparation of a next candidate from accidentally reusing entropy when the current round accepts. Speculatively generated candidates may remain buffered for a later logically fresh attempt, but their random segments are consumed exactly once.

\subsection{Interface and Integration Boundary}

The pair engine follows the task and memory-facing conventions of the open FalconSign sampler implementation so that it can replace the leaf-level sampler without changing the surrounding FFS schedule~\cite{falconsign}. The externally visible transaction contains two centers, one shared standard deviation, and two returned integer samples. Assistance is internal: the caller neither observes which physical lane produced the unfinished sample nor issues an additional request.

This interface choice also defines the comparison boundary used in Section~\ref{sec:evaluation}. Bi-SamplerZ includes its ChaCha20 and refill logic; sampler figures from architectures that rely on external random-number or software services are therefore informative but not directly interchangeable.

\section{Secondary Arithmetic and Logic Optimizations}
\label{sec:microopt}

The cooperative scheduler defines the main architectural framework, while the implemented engine also incorporates several arithmetic and logic optimizations that reduce the cost of each sampling attempt. These optimizations are retained from the original hardware design and complement the pair-level scheduling mechanism.

\subsection{Bounded Fixed-Point Internal Datapath}

The prototype represents bounded internal quantities with an 81-bit fixed-point word containing nine integer/magnitude bits and 72 fractional bits. Sign is carried separately where convenient. The width follows the range of the candidate-dependent term. Since $z_0\in\{0,\ldots,18\}$, $z\in[-18,19]$, and $r\in[0,1)$,
\begin{equation}
|z-r|<19,\qquad (z-r)^2<361<2^9.
\label{eq:range}
\end{equation}
Nine integer bits therefore cover the dominant intermediate magnitude, while 72 fractional bits provide a fine implementation grid for these bounded quantities. This range-and-resolution argument characterizes the prototype datapath width; it does not by itself establish bit-exact equivalence to binary64 Falcon arithmetic. In particular, cancellation and very small fractional centers require implementation-specific numerical analysis, and recent work has shown that Falcon-related samplers can be sensitive to small floating-point discrepancies~\cite{lin2025,twfalcon2026,braga2026}. We therefore treat the fixed-point path as the arithmetic realization used by the evaluated prototype and keep the cooperative scheduling theorem relative to that implemented baseline distribution $\Dimpl$. The same assistance controller can be attached to a standards-conforming or separately verified arithmetic path without changing its pair-level logic.

\subsection{Precomputed Small-Domain Terms}

The second term of~\eqref{eq:xtable} depends only on $z_0$ and the design-time constant $\sigma_{\max}$, so the 19 values
\begin{equation}
T[k]=\frac{k^2}{2\sigma_{\max}^2},\qquad k=0,\ldots,18,
\end{equation}
are stored in a small lookup table. This removes a square and constant division from every attempt. The integer offsets in the final $z+\lfloor\mu\rfloor$ operation similarly occupy a small domain, allowing pre-encoded arithmetic values to replace a general integer-to-real conversion block in the implemented datapath.

\subsection{Transition-Based BaseSampler Selector}

A direct RCDT implementation compares one uniform word with all thresholds and sums the comparison bits. Because the thresholds are monotone, the comparison vector has the form $11\ldots1100\ldots00$. The sampled magnitude is therefore the transition position rather than an arbitrary population count. Figure~\ref{fig:Tri-state Gates Logics} shows the implemented transformation.

\begin{figure}[!t]
\centering
\includegraphics[width=0.48\textwidth]{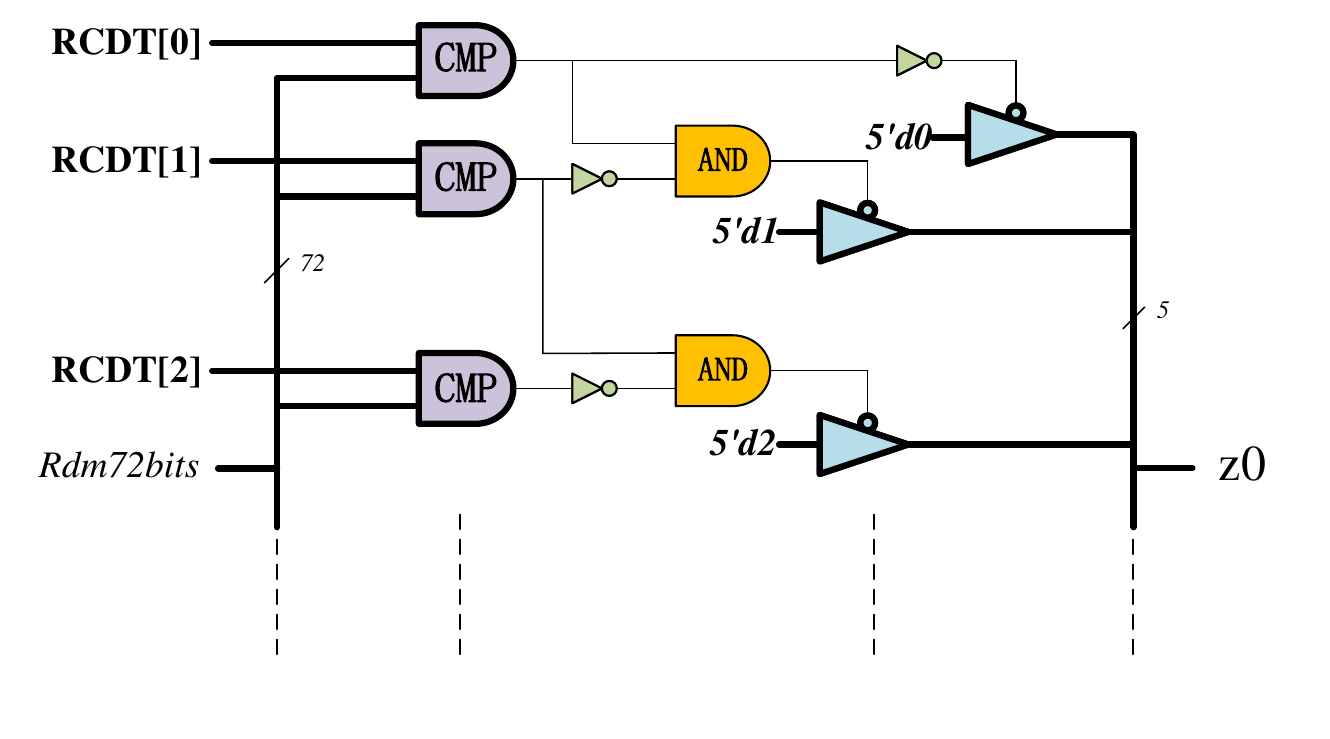}
\caption{Tri-state-gate-based logic for the optimized BaseSampler.}
\label{fig:Tri-state Gates Logics}
\end{figure}

For ASIC synthesis, the one-hot transition can drive a compact index-selection network; the original implementation permits a tri-state realization when supported by the cell methodology. For FPGA synthesis, where internal tri-states are not a portable fabric primitive, the same monotone transition is mapped onto FPGA-compatible encoder/multiplexer logic. Thus, the tri-state form is an ASIC mapping option rather than a requirement of the architecture. The optimization changes only the realization of the RCDT lookup and not its threshold values or sampled distribution.

\section{Experimental Evaluation}
\label{sec:evaluation}

\subsection{Implementation Setup}

We evaluate Bi-SamplerZ at the hardware boundary of a paired leaf-level sampling transaction. The evaluation has two goals: first, to quantify the complete engine against representative Falcon sampler baselines; and second, to isolate the additional benefit contributed by rejection-aware assistance once the dual-datapath engine is already present.

The FPGA implementation was synthesized with Xilinx Vivado 2019.2 for the ZCU104 platform at a 150~MHz target. The ASIC implementation was synthesized with Synopsys Design Compiler using a GlobalFoundries 22-nm library with a 1.5~GHz synthesis target. These ASIC numbers are pre-layout synthesis estimates. The reported ASIC power is likewise a synthesis-level estimate and is not used for cross-paper energy-efficiency ranking because activity assumptions and implementation stages are not uniform across the compared works. The project RTL is available at \url{https://github.com/Shaibk/Bi-SamplerZ}. Table~\ref{tab:implementation} summarizes the implementation totals.

\begin{table}[t]
\caption{Bi-SamplerZ Implementation Summary}
\label{tab:implementation}\label{tab:resource_power}
\centering
\setlength{\tabcolsep}{3.2pt}
\renewcommand{\arraystretch}{1.10}
\begin{tabular}{lcc}
\toprule
Metric & FPGA & ASIC \\
\midrule
Target & Xilinx ZCU104 & GF 22 nm \\
Tool & Vivado 2019.2 & Synopsys DC \\
Synthesis target & 150 MHz & 1.5 GHz \\
LUTs & 14,327 & -- \\
Flip-flops & 10,841 & -- \\
DSPs & 85 & -- \\
Cell area & -- & 68,700 $\mu$m$^2$ \\
Power & -- & 40.993 mW \\
\bottomrule
\end{tabular}
\end{table}

ChaCha20 and the wide multipliers dominate the implementation hierarchy. The cooperative framework reuses the two datapaths already provisioned for pair execution; its additional structures are primarily target-selection multiplexing, tagged state, and the switching states required to retarget a completed lane.

\subsection{Cycle Structure}

The pair-level preprocessing stage requires 19 cycles, the intermediate registration stage requires one cycle, and one rejection-evaluation pass occupies 34--41 cycles depending on the byte-wise comparison termination. Following an asymmetric normal round, \texttt{SWITCHL}/\texttt{SWITCHR} requires nine cycles before both physical paths can operate on the unfinished target. If neither initial request rejects, the complete pair is returned in 59 cycles.

The expected-latency model replays the failure distribution in Table~\ref{tab:failstats} with the measured state costs. It is a software timing model driven by published rejection statistics, not a new statistical characterization of Falcon's complete input space. The same model and cycle costs are used for the assistance-enabled and assistance-disabled Bi-SamplerZ variants, which makes their difference a controlled architectural ablation.

\subsection{Assistance Ablation}

Table~\ref{tab:ablation} gives the central result. Disabling assistance leaves the shared front end and the two normal rejection datapaths unchanged; after an asymmetric outcome, only the unfinished lane continues. Enabling assistance adds the nine-cycle switch and then uses two lanes for the remaining target.

\begin{table}[t]
\caption{Expected Pair Latency and Assistance Ablation}
\label{tab:ablation}\label{tab:latency-comparison}
\centering
\setlength{\tabcolsep}{4.5pt}
\renewcommand{\arraystretch}{1.12}
\begin{tabular}{lrr}
\toprule
Design point & Expected cycles & Change \\
\midrule
Bi-SamplerZ without assistance & 111.54 & baseline \\
Bi-SamplerZ with assistance & 106.08 & $-5.46$ \\
\midrule
Relative reduction & \multicolumn{2}{c}{4.90\%} \\
\bottomrule
\end{tabular}
\end{table}

Despite the nine-cycle retargeting overhead, assistance reduces the expected pair latency by 5.46 cycles, corresponding to an additional 4.90\% reduction over the same Bi-SamplerZ datapath with static post-acceptance assignment. This result directly validates the proposed cooperation mechanism: after the hardware has exposed pair-level parallelism, reusing the completed lane shortens the expected completion tail rather than leaving that lane idle. The measured trend is consistent with the threshold analysis in~\eqref{eq:threshold}.

This 4.90\% figure is an assistance-only ablation at the paired-sampler boundary. The larger complete-architecture improvement reported next additionally includes the dual rejection pipelines, shared preprocessing, arithmetic organization, and pipeline overlap.

\subsection{Direct Baseline Comparison}

Table~\ref{tab:baseline} compares Bi-SamplerZ with the evaluated FalconSign sampler at the same logical boundary: hardware required to produce the two \SamplerZ\ outputs of one scalar FFS leaf. FalconSign reports 74 cycles for a single sample and 137 cycles for double sampling without rejection~\cite{falconsign}; therefore, the 137-cycle baseline used here is the reported two-sample latency and is not obtained by serializing or doubling a single-call latency. The FalconSign RTL was re-synthesized in the local FPGA and ASIC flows used by the original implementation study. This provides a concrete architecture-to-architecture baseline while preserving each design's frequency target and resource organization.

\begin{table*}[t]
\centering
\begin{threeparttable}
\caption{Direct Two-Sample Hardware Comparison with the Evaluated FalconSign Baseline}
\label{tab:baseline}\label{tab:related-works}
\setlength{\tabcolsep}{5.5pt}
\renewcommand{\arraystretch}{1.12}
\begin{tabular}{llcccccc}
\toprule
Platform & Design & Resources / area & Freq. & \makecell{No-reject\\cycles} & \makecell{No-reject\\latency} & \makecell{Expected\\cycles} & \makecell{Expected\\latency} \\
\midrule
FPGA & Bi-SamplerZ & 14,327 LUT / 10,841 FF / 85 DSP & 150 MHz & 59 & 393.3 ns & 106.08 & 707.2 ns \\
FPGA & FalconSign~\cite{falconsign} & 14,710 LUT / 10,731 FF / 76 DSP & 185 MHz & 137 & 740.5 ns & 230.83 & 1247.7 ns \\
\midrule
ASIC & Bi-SamplerZ & 68,700 $\mu$m$^2$ & 1.5 GHz & 59 & 39.3 ns & 106.08 & 70.7 ns \\
ASIC & FalconSign~\cite{falconsign} & 58,426 $\mu$m$^2$ & 500 MHz & 137 & 274.0 ns & 230.83 & 461.7 ns \\
\bottomrule
\end{tabular}
\begin{tablenotes}
\footnotesize
\item Expected-cycle values use the same rejection histogram and timing-model methodology adopted by the original Bi-SamplerZ evaluation. FalconSign's 137-cycle entry is its reported double-sampling latency without rejection, not twice a 137-cycle single-sample latency~\cite{falconsign}. The ASIC latency values derived from target frequency are synthesis-level estimates. The table does not compare against later arithmetic-focused implementations whose boundaries and platforms differ.
\end{tablenotes}
\end{threeparttable}
\end{table*}

At the FPGA target, Bi-SamplerZ uses 14,327 LUTs, 10,841 flip-flops, and 85 DSPs, compared with 14,710 LUTs, 10,731 flip-flops, and 76 DSPs for the evaluated FalconSign sampler. The no-rejection pair latency decreases from 740.5~ns to 393.3~ns, corresponding to a 1.88$\times$ latency speedup at the respective frequency targets. Under the common rejection model, expected pair latency decreases from 1247.7~ns to 707.2~ns (1.76$\times$). On ASIC, the evaluated Bi-SamplerZ synthesis point uses 17.6\% more cell area than the re-synthesized FalconSign sampler and uses a 1.5~GHz target rather than 500~MHz; the corresponding frequency-derived no-rejection and modeled expected pair latencies are 6.97$\times$ and 6.53$\times$ lower, respectively. These ASIC latency ratios should be interpreted as synthesis-level comparison points rather than post-layout silicon speedups.

In cycle terms, which isolate scheduling and pipeline organization from clock-frequency differences, the complete Bi-SamplerZ engine reduces the modeled expected two-sample cost from 230.83 to 106.08 cycles, a 54.0\% reduction relative to the evaluated FalconSign baseline. The rejection-free pair requires 59 cycles instead of 137, a 56.9\% reduction. These are complete-engine results; the additional gain attributable specifically to cooperative assistance is the 4.90\% ablation reported in Table~\ref{tab:ablation}.

For context, Table~\ref{tab:broader} retains the representative FPGA comparisons used in the original implementation study. The table reports latency explicitly rather than assigning a global ``best'' label, because later Falcon samplers use different arithmetic organizations and move the absolute performance frontier~\cite{pendyala2025,vojacek2025,falconwings2026}. The comparison nevertheless shows that Bi-SamplerZ's pair-oriented implementation provides substantial latency reduction relative to the evaluated 2024-era hardware and HW/SW baselines.

\begin{table*}[t]
\centering
\caption{Representative FPGA Comparison for Two SamplerZ Outputs}
\label{tab:broader}
\setlength{\tabcolsep}{5.2pt}
\renewcommand{\arraystretch}{1.12}
\begin{tabular}{lcccccc}
\toprule
Design & LUTs & FFs & DSPs & Freq. (MHz) & Cycles & Latency ($\mu$s) \\
\midrule
\textbf{Bi-SamplerZ} & 14,327 & 10,841 & 85 & 150 & 59 & 0.393 \\
FalconSign~\cite{falconsign} & 14,710 & 10,731 & 76 & 185 & 137 & 0.741 \\
RVCE-FAL~\cite{Yu} & 4,761 & 692 & 9 & 83 & 208 & 2.506 \\
AHS Base~\cite{Karabulut} & 3,456 & 732 & 116 & 117 & 3,382 & 28.906 \\
AHS Perf. Opt.~\cite{Karabulut} & 2,523 & 2,028 & 9 & 117 & 3,452 & 29.504 \\
AHS Area Opt.~\cite{Karabulut} & 2,328 & 1,030 & 3 & 117 & 3,614 & 30.889 \\
\bottomrule
\end{tabular}
\end{table*}

The evaluation is intentionally reported at the sampler boundary because that is where the proposed framework acts and where the available RTL baselines can be compared consistently. End-to-end signing speedup additionally depends on the surrounding FFS schedule, buffering, and other Falcon kernels; we therefore do not extrapolate the pair-level numbers into an unmeasured full-signer throughput claim. Software-only cycle counts are also not used as a primary quantitative baseline because they conflate processor microarchitecture, vector ISA support, arithmetic libraries, and random-number services with the transaction-level scheduling effect studied here.

\section{Related Work and Positioning}
\label{sec:related}

\subsection{Complete and Programmable Falcon Accelerators}

Schmid \emph{et al.} present an HLS implementation of Falcon signing and key generation and transform recursive software into hardware-compatible iterative control~\cite{Micheal}. EFX coordinates hardware and software for low-end embedded platforms and contributes fine-grained scheduling together with the failure measurements used in this paper~\cite{Lee}. RVCE-FAL exposes custom RISC-V scalar-vector operations for Falcon kernels~\cite{Yu}; this direction is consistent with broader efforts to extend RISC-V vector execution for post-quantum workloads~\cite{risc}. At the system level, prior work has also investigated hardware acceleration of both Dilithium and Falcon~\cite{loT}, while RISC-V-assisted signature verification demonstrates a complementary design point that does not exercise Falcon's Gaussian signing sampler~\cite{verify}. These architectures trade dedicated sampler specialization for broader programmability, system integration, or acceleration of different signature phases.

FalconSign is a complete FPGA/ASIC signing accelerator with a new FFS schedule, parallel floating-point resources, configurable randomness generation, and an optimized \SamplerZ~\cite{falconsign}. Bi-SamplerZ uses FalconSign as a compatible engineering reference while addressing a different architectural level: the dynamic mapping of the two logical leaf samples onto physical rejection datapaths after one request completes. Complementary work by Alsuhli \emph{et al.} optimizes Falcon's FFT/IFFT processing, illustrating that a full signer must coordinate several expensive kernels beyond \SamplerZ~\cite{alsuhli2025}.

\subsection{Sampler-Specific Parallelism}

Karabulut and Aysu develop a hardware/software co-design for Falcon's complete discrete Gaussian sampling procedure, including dedicated arithmetic and multiple hardware configurations~\cite{Karabulut}. Pendyala \emph{et al.} subsequently improve individual-attempt throughput through datapath-aware floating-point pipelines, parallel sampling loops, loop unrolling, an Estrin polynomial evaluator, and constant-latency \BerExp~\cite{pendyala2025}. A 2025 bachelor thesis by Voj\'a\v{c}ek reports a compact implementation that produces the required two parallel samples with low average latency~\cite{vojacek2025}; we cite it as non-peer-reviewed contemporaneous evidence rather than as a journal baseline.

FALCON-Wings, accepted for A-SSCC 2026, reports a prediction-based dual-channel sampler within a full accelerator~\cite{falconwings2026}. Its appearance makes the distinction between \emph{two channels} and the present mechanism important. Bi-SamplerZ does not claim novelty from lane count. Its defining event is an \emph{observed asymmetric outcome}: after one logical request has completed, that physical path is reassigned to issue an additional independent attempt for the other request. Prediction and post-completion reassignment are different scheduling policies and may be complementary.

Table~\ref{tab:positioning} summarizes the mechanism emphasized by each representative design without attempting to normalize incomparable implementation results.

\begin{table*}[t]
\centering
\caption{Positioning Relative to Representative Falcon Hardware Work}
\label{tab:positioning}
\setlength{\tabcolsep}{5.2pt}
\renewcommand{\arraystretch}{1.12}
\begin{tabular}{>{\raggedright\arraybackslash}p{0.16\textwidth}>{\raggedright\arraybackslash}p{0.12\textwidth}>{\raggedright\arraybackslash}p{0.30\textwidth}>{\raggedright\arraybackslash}p{0.32\textwidth}}
\toprule
Work & Integration & Main sampler/architecture emphasis & Relation to multiple sampling attempts \\
\midrule
Karabulut--Aysu~\cite{Karabulut} & HW/SW & Dedicated Gaussian-sampling offload and configurable hardware datapaths & Parallel hardware configurations coordinated with software \\
Lee \emph{et al.}~\cite{Lee} & HW/SW & Fine-grained embedded scheduling and resource-conscious acceleration & Concurrent task scheduling; reports rejection histogram \\
RVCE-FAL~\cite{Yu} & vector co-processor & Programmable scalar-vector extensions for Falcon kernels & Vectorized execution rather than a dedicated pair transaction \\
FalconSign~\cite{falconsign} & full hardware & Complete signer, FFS scheduling, FP resource reuse, optimized \SamplerZ & System-level scheduling and dedicated sampler pipeline \\
Pendyala \emph{et al.}~\cite{pendyala2025} & full hardware & Fast attempt arithmetic, parallel loops, redesigned exponential path & Multiple loops improve throughput of rejection attempts \\
Voj\'a\v{c}ek~\cite{vojacek2025} & full hardware thesis & Compact low-latency two-sample \SamplerZ & Two parallel samples with reduced random-bit usage \\
FALCON-Wings~\cite{falconwings2026} & full hardware & Fine-grained signer and prediction-based dual-channel sampler & Dual channels directed by rejection prediction \\
\textbf{Bi-SamplerZ} & full hardware & Shared-context pair engine and rejection-aware controller & Completed lane is dynamically reassigned after observed asymmetry \\
\bottomrule
\end{tabular}
\end{table*}

\subsection{Numerical Correctness Context}

Peikert's work and subsequent lattice-sampling literature establish multiple approaches to efficient discrete Gaussian sampling~\cite{CDTsampling}. Generic hardware Gaussian-noise generators address a different design target: they optimize a distribution-generation primitive independently of Falcon, whereas Falcon requires its protocol-defined, run-time-parameterized discrete sampler and Bernoulli rejection semantics~\cite{Falcon}. We therefore do not treat generic noise generators as interchangeable Falcon \SamplerZ\ baselines. Embedded implementations of standard lattice primitives further illustrate the recurring cost of realizing distribution sampling and arithmetic under tight hardware constraints~\cite{frodo}.

Falcon's specific FFS and scalar rejection sampler combine variable parameters with precision-sensitive computations. Lin \emph{et al.} show that small discrepancies can have serious consequences in certain repeated-input or derandomized Falcon settings~\cite{lin2025}. TWFalcon develops triple-word arithmetic with formally motivated precision guarantees, while Braga \emph{et al.} give a secure fixed-point implementation methodology~\cite{twfalcon2026,braga2026}. Fouque \emph{et al.} revisit Falcon's formal security and analyze a modified, provable formulation incorporated into the forthcoming standardization direction~\cite{fouque2026}. These developments reinforce two choices in this paper: we avoid asserting that an unpublished final FN-DSA is identical to the 2020 Falcon submission, and we state scheduling correctness relative to a fixed implemented sampler distribution rather than deriving a new numerical-security claim from word width alone. Accordingly, statistical goodness-of-fit of a separately chosen arithmetic realization is an arithmetic-conformance question; the contribution evaluated here is that cooperative scheduling preserves whatever fixed baseline distribution that realization implements.

\subsection{Architectural Position}

Existing work improves two largely orthogonal dimensions. \emph{Attempt-level} techniques reduce the cost of computing one proposal and one acceptance decision through faster arithmetic, unrolling, vectorization, or prediction. \emph{Transaction-level} techniques determine how logical samples are mapped onto physical resources as their completion states evolve. Bi-SamplerZ introduces a transaction-level cooperative framework. Faster arithmetic changes $L$ in~\eqref{eq:threshold}, but it does not eliminate the asymmetric-completion event that motivates assistance. Conversely, the Bi-SamplerZ scheduler can surround different arithmetic cores as long as helper attempts preserve the same target distribution and use independent randomness. This makes the proposed mechanism complementary to the newer arithmetic-focused and prediction-based Falcon accelerators rather than dependent on a particular performance ranking.

\section{Conclusion}
\label{sec:conclusion}

This paper presents Bi-SamplerZ, a rejection-aware cooperative sampling framework for Falcon signature hardware. The key observation is that Falcon's scalar FFS leaves naturally issue two independent \SamplerZ\ requests with a shared standard deviation, while rejection sampling makes their completion times unequal. A static dual-path sampler exposes the pair-level parallelism but wastes a datapath after an asymmetric outcome. Bi-SamplerZ instead freezes the completed logical result and dynamically retargets the released physical path so that both lanes cooperate on the unfinished request.
The mechanism is a scheduling transformation rather than a change to the sampler itself. With fresh non-overlapping randomness and value-independent selection among simultaneous helper acceptances, the cooperative execution preserves the implemented distribution of the two logical sampler outputs. The framework is realized with shared preprocessing, two rejection pipelines, explicit random-word allocation, and the arithmetic optimizations of the original hardware engine.
The existing FPGA and ASIC results demonstrate both the complete architecture and the contribution of cooperation. Relative to the evaluated FalconSign sampler, Bi-SamplerZ reduces the rejection-free pair cycle count from 137 to 59 and the modeled expected pair cycle count from 230.83 to 106.08, corresponding to reductions of 56.9\% and 54.0\%, respectively. Within the Bi-SamplerZ architecture, enabling assistance further reduces expected pair latency from 111.54 to 106.08 cycles despite a nine-cycle switching cost. These results establish rejection-induced idle capacity as an exploitable architectural resource.
Recent Falcon implementations have advanced the absolute latency frontier through faster arithmetic, parallel loops, compact datapaths, and prediction-based dual-channel designs. Bi-SamplerZ addresses a complementary question: how to schedule already-provisioned parallel hardware after one logical rejection loop finishes before the other. The proposed pair-aware cooperative framework, therefore, remains applicable as the underlying Falcon/FN-DSA arithmetic cores continue to evolve. Post-layout timing validation, arithmetic conformance to the final FN-DSA profile, and side-channel countermeasures remain integration-level tasks beyond the scheduling contribution evaluated here.

\section*{Acknowledgment}
This research was supported by Khalifa University Center for Intelligent Computing, Networks and Sensing (ICONS), Khalifa University of Science and Technology (KU-ICONS), and partially supported by the British University in Dubai (BUiD) research project INF027.

\bibliographystyle{IEEEtran}
\bibliography{main}

@article{falconsign,
  author  = {Yi Ouyang and Yihong Zhu and Wenping Zhu and Bohan Yang and Zirui Zhang and Hanning Wang and Qichao Tao and Min Zhu and Shaojun Wei and Leibo Liu},
  title   = {{FalconSign}: An Efficient and High-Throughput Hardware Architecture for {Falcon} Signature Generation},
  journal = {IACR Transactions on Cryptographic Hardware and Embedded Systems},
  volume  = {2025},
  number  = {1},
  pages   = {203--226},
  year    = {2025},
  doi     = {10.46586/tches.v2025.i1.203-226},
  url     = {https://tches.iacr.org/index.php/TCHES/article/view/11927}
}

@ARTICLE{Lee,
  author={Lee, Yongseok and Youn, Jonghee and Nam, Kevin and Jung, Heon Hui and Cho, Myunghyun and Na, Jimyung and Park, Jong-Yeon and Jeon, Seungsu and Kang, Bo Gyeong and Oh, Hyunyoung and Paek, Yunheung},
  journal={IEEE Access}, 
  title={An Efficient Hardware/Software Co-Design for FALCON on Low-End Embedded Systems}, 
  year={2024},
  volume={12},
  number={},
  pages={57947-57958},
  doi={10.1109/ACCESS.2024.3387489}}

@INPROCEEDINGS{Karabulut,
  author={Karabulut, Emre and Aysu, Aydin},
  booktitle={2024 IEEE International Symposium on Hardware Oriented Security and Trust (HOST)}, 
  title={A Hardware-Software Co-Design for the Discrete Gaussian Sampling of FALCON Digital Signature}, 
  year={2024},
  volume={},
  number={},
  pages={90-100},
  doi={10.1109/HOST55342.2024.10545399}}

@INPROCEEDINGS{Yu,
  author={Yu, Xinglong and Sun, Yi and Zhao, Yifan and Kuang, Honglin and Han, Jun},
  booktitle={2024 Design, Automation \& Test in Europe Conference \& Exhibition (DATE)}, 
  title={RVCE-FAL: A RISC-V Scalar-Vector Custom Extension for Faster FALCON Digital Signature}, 
  year={2024},
  volume={},
  number={},
  pages={1-6},
  doi={10.23919/DATE58400.2024.10546713}}

@misc{Falcon,
  author       = {Pierre-Alain Fouque and Jeffrey Hoffstein and Paul Kirchner and Vadim Lyubashevsky and Thomas Pornin and Thomas Prest and Thomas Ricosset and Gregor Seiler and William Whyte and Zhenfei Zhang},
  title        = {{Falcon}: Fast-Fourier Lattice-Based Compact Signatures over {NTRU}},
  howpublished = {Specification submitted to the NIST Post-Quantum Cryptography Standardization Process},
  year         = {2020},
  url          = {https://falcon-sign.info/falcon.pdf}
}

@inproceedings{CDTsampling,
  author    = {Chris Peikert},
  title     = {An Efficient and Parallel Gaussian Sampler for Lattices},
  booktitle = {Advances in Cryptology--CRYPTO 2010},
  series    = {Lecture Notes in Computer Science},
  volume    = {6223},
  pages     = {80--97},
  publisher = {Springer},
  year      = {2010},
  doi       = {10.1007/978-3-642-14623-7_5}
}

@misc{Micheal,
  author       = {Michael Schmid and Dorian Amiet and Jan Wendler and Paul Zbinden and Tao Wei},
  title        = {Falcon Takes Off---A Hardware Implementation of the {Falcon} Signature Scheme},
  howpublished = {Cryptology ePrint Archive, Paper 2023/1885},
  year         = {2023},
  url          = {https://eprint.iacr.org/2023/1885}
}

@inproceedings{qtesla,
  title={Merge-exchange sort based discrete Gaussian sampler with fixed memory access pattern},
  author={Tian, Shanquan and Wang, Wen and Szefer, Jakub},
  booktitle={2019 International Conference on Field-Programmable Technology (ICFPT)},
  pages={126--134},
  year={2019},
  organization={IEEE}
}

@article{blss,
  title={On practical discrete Gaussian samplers for lattice-based cryptography},
  author={Howe, James and Khalid, Ayesha and Rafferty, Ciara and Regazzoni, Francesco and O'Neill, M{\'a}ire},
  journal={IEEE Transactions on Computers},
  volume={67},
  number={3},
  pages={322--334},
  year={2016},
  publisher={IEEE}
}

@inproceedings{lp2,
  title={Towards efficient discrete Gaussian sampling for lattice-based cryptography},
  author={Du, Chaohui and Bai, Guoqiang},
  booktitle={2015 25th International Conference on Field Programmable Logic and Applications (FPL)},
  pages={1--6},
  year={2015},
  organization={IEEE}
}

@article{frodo,
  title={Standard lattice-based key encapsulation on embedded devices},
  author={Howe, James and Oder, Tobias and Krausz, Markus and G{\"u}neysu, Tim},
  journal={Cryptology ePrint Archive},
  year={2018}
}

@article{design-time,
  title={Efficient, flexible, and constant-time gaussian sampling hardware for lattice cryptography},
  author={Karabulut, Emre and Alkim, Erdem and Aysu, Aydin},
  journal={IEEE Transactions on Computers},
  volume={71},
  number={8},
  pages={1810--1823},
  year={2021},
  publisher={IEEE}
}

@article{verify,
  title={Post-quantum signatures on RISC-V with hardware acceleration},
  author={Karl, Patrick and Schupp, Jonas and Fritzmann, Tim and Sigl, Georg},
  journal={ACM Transactions on Embedded Computing Systems},
  volume={23},
  number={2},
  pages={1--23},
  year={2024},
  publisher={ACM New York, NY}
}

@article{loT,
  title={Hardware accelerators for digital signature algorithms dilithium and falcon},
  author={Beckwith, Luke and Nguyen, Duc Tri and Gaj, Kris},
  journal={IEEE Design \& Test},
  year={2023},
  publisher={IEEE}
}

@inproceedings{risc,
  title={Enhancing risc-v vector extension for efficient application of post-quantum cryptography},
  author={Zhao, Yifan and Kuang, Honglin and Sun, Yi and Yang, Zhen and Chen, Chen and Meng, Jianyi and Han, Jun},
  booktitle={2023 IEEE 34th International Conference on Application-specific Systems, Architectures and Processors (ASAP)},
  pages={10--17},
  year={2023},
  organization={IEEE}
}

@misc{nist_selected,
  author       = {{National Institute of Standards and Technology}},
  title        = {Post-Quantum Cryptography: Selected Algorithms},
  howpublished = {Computer Security Resource Center},
  year         = {2022},
  url          = {https://csrc.nist.gov/Projects/post-quantum-cryptography/post-quantum-cryptography-standardization/selected-algorithms},
  note         = {Updated December 2025; accessed September 2026}
}

@misc{nist_fips206,
  author       = {Ray Perlner},
  title        = {{FIPS 206: FN-DSA (Falcon)}},
  howpublished = {Sixth NIST Post-Quantum Cryptography Standardization Conference},
  month        = sep,
  year         = {2025},
  url          = {https://csrc.nist.gov/Presentations/2025/fips-206-fn-dsa-falcon}
}

@inproceedings{gpv2008,
  author    = {Craig Gentry and Chris Peikert and Vinod Vaikuntanathan},
  title     = {Trapdoors for Hard Lattices and New Cryptographic Constructions},
  booktitle = {Proceedings of the 40th Annual ACM Symposium on Theory of Computing},
  pages     = {197--206},
  year      = {2008},
  doi       = {10.1145/1374376.1374407}
}

@article{alsuhli2025,
  author  = {Ghada Alsuhli and Hani Saleh and Mahmoud Al-Qutayri and Baker Mohammad and Thanos Stouraitis},
  title   = {Area and Power Efficient {FFT/IFFT} Processor for {FALCON} Post-Quantum Cryptography},
  journal = {IEEE Transactions on Emerging Topics in Computing},
  volume  = {13},
  number  = {2},
  pages   = {423--437},
  year    = {2025},
  doi     = {10.1109/TETC.2024.3407124}
}

@misc{pendyala2025,
  author       = {Sharath Pendyala and Rahul Magesh and Elif Bilge Kavun and Aydin Aysu},
  title        = {Outrunning the Millennium {FALCON}: Speed Records for {FALCON} on {FPGAs}},
  howpublished = {Cryptology ePrint Archive, Paper 2025/1490},
  year         = {2025},
  url          = {https://eprint.iacr.org/2025/1490}
}

@mastersthesis{vojacek2025,
  author  = {Michal Voj{\'a}{\v c}ek},
  title   = {Efficient {Falcon} Signature Generation in Hardware},
  school  = {Czech Technical University in Prague},
  type    = {Bachelor thesis},
  year    = {2025},
  url     = {https://dspace.cvut.cz/entities/publication/9cdad256-5ce7-4b9f-b18f-c8b7fc2592b0}
}

@inproceedings{falconwings2026,
  author    = {Mingbo Wang and Aobo Li and Tianze Huang and Yibo Wang and Dongsheng Liu and Xuecheng Zou and Jiahao Lu},
  title     = {{FALCON-Wings}: A Fine-Grained Hardware Accelerator for {Falcon} Signature Generation with an Efficient Dual Channel Sampler and Versatile Computing Fabric},
  booktitle = {2026 IEEE Asian Solid-State Circuits Conference (A-SSCC)},
  note      = {Accepted for lecture presentation},
  year      = {2026},
  url       = {https://epapers2.org/asscc2026/ESR/paper_details.php?paper_id=1391}
}

@inproceedings{lin2025,
  author    = {Xiuhan Lin and Mehdi Tibouchi and Yang Yu and Shiduo Zhang},
  title     = {Do Not Disturb a Sleeping {Falcon}: Floating-Point Error Sensitivity of the {Falcon} Sampler and Its Consequences},
  booktitle = {Advances in Cryptology--EUROCRYPT 2025},
  series    = {Lecture Notes in Computer Science},
  volume    = {15602},
  pages     = {213--244},
  publisher = {Springer},
  year      = {2025},
  doi       = {10.1007/978-3-031-91124-8_8}
}

@inproceedings{fouque2026,
  author    = {Pierre-Alain Fouque and Phillip Gajland and Hubert de Groote and Jonas Janneck and Eike Kiltz},
  title     = {A Closer Look at {Falcon}},
  booktitle = {Advances in Cryptology--EUROCRYPT 2026},
  series    = {Lecture Notes in Computer Science},
  volume    = {16544},
  pages     = {3--31},
  publisher = {Springer},
  year      = {2026},
  doi       = {10.1007/978-3-032-25327-9_1}
}

@article{twfalcon2026,
  author  = {Stef Halmans and Christine van Vredendaal and Tobias Schneider and Frank Custers and Tim G{\"u}neysu},
  title   = {{TWFalcon}: Triple-Word Arithmetic for {Falcon}: Giving {Falcon} the Precision to Fly Securely},
  journal = {IACR Transactions on Cryptographic Hardware and Embedded Systems},
  volume  = {2026},
  number  = {2},
  pages   = {347--371},
  year    = {2026},
  doi     = {10.46586/tches.v2026.i2.347-371}
}

@inproceedings{braga2026,
  author    = {Daniel De Almeida Braga and Pierre-Alain Fouque and Bachir Lachguel and Thomas Prest},
  title     = {Toward a Secure Fixed-Point Implementation of the {Falcon} Signature Scheme},
  booktitle = {Advances in Cryptology--CRYPTO 2026},
  series    = {Lecture Notes in Computer Science},
  volume    = {16802},
  pages     = {166--197},
  publisher = {Springer},
  year      = {2026},
  doi       = {10.1007/978-3-032-35377-1_6}
}

@article{yan2026lattice,
  title={Lattice-based cryptographic accelerators for the post-quantum era: Architectures, optimizations, and implementation challenges},
  author={Yan, Hua and Wu, Lei and Sun, Qiming and He, Pengzhou},
  journal={Electronics},
  volume={15},
  number={2},
  pages={475},
  year={2026},
  publisher={MDPI}
}

@article{smole2026giving,
  title={Giving {Falcon} Its Wings: Fast Hardware Acceleration for Signing and Verification Using {FN-DSA}},
  author={Smole, Vid and Aikata, Aikata and Pilipovi{\'c}, Ratko and Roy, Sujoy Sinha},
  journal={IEEE Transactions on Circuits and Systems I: Regular Papers},
  year={2026},
  publisher={IEEE}
}

@article{sakwa2026survey,
  title={Survey on post-quantum cryptography implementations and deployment challenges},
  author={Sakwa, Cyprian Omukhwaya and Li, Fagen},
  journal={Computer Science Review},
  volume={61},
  pages={100975},
  year={2026},
  publisher={Elsevier}
}

@article{alshehhi2026design,
  title={Design and implementation of a fixed-point {FFT/IFFT} processor for {NTRU} lattices},
  author={Alshehhi, Muna and Alsuhli, Ghada and Saleh, Hani and Stouraitis, Thanos and Caullery, Florian},
  journal={Microprocessors and Microsystems},
  pages={105282},
  year={2026},
  publisher={Elsevier}
}

@article{alsuhli2025efficient,
  title={Efficient {NTT/INTT} processor for {FALCON} post-quantum cryptography},
  author={Alsuhli, Ghada and Saleh, Hani and Al-Qutayri, Mahmoud and Mohammad, Baker and Stouraitis, Thanos},
  journal={Journal of Information Security and Applications},
  volume={93},
  pages={104177},
  year={2025},
  publisher={Elsevier}
}

@article{luo2026rethinking,
  title={Rethinking Area Optimization for High-Throughput Hardware Accelerator of Number Theoretic Transform: A Case Study on {ML-KEM}},
  author={Luo, Hao and Li, Minghao and Song, Suwen and Wang, Zhongfeng},
  journal={IEEE Transactions on Very Large Scale Integration (VLSI) Systems},
  year={2026},
  publisher={IEEE}
}

@article{he2021novel,
  title={Novel low-complexity polynomial multiplication over hybrid fields for efficient implementation of binary ring-{LWE} post-quantum cryptography},
  author={He, Pengzhou and Guin, Ujjwal and Xie, Jiafeng},
  journal={IEEE Journal on Emerging and Selected Topics in Circuits and Systems},
  volume={11},
  number={2},
  year={2021}
}

\end{document}